\documentclass[12pt]{article}

\usepackage[authoryear,round]{natbib}
\usepackage{amsfonts}
\usepackage{amssymb}
\usepackage{amsthm}
\usepackage{graphics}
\usepackage{graphicx}
\usepackage{rotating}
\usepackage{subfigure}
\usepackage{booktabs}
\usepackage{algorithm}
\usepackage{algorithmic}
\usepackage{color}
\usepackage[rgb]{xcolor}
\usepackage{amsmath}
\usepackage{mathrsfs}
\usepackage{hyperref}
\usepackage{setspace}
\usepackage{grffile}
\usepackage{multirow}
\usepackage{verbatim}   % useful for program listings
\usepackage{url}
\usepackage{multirow}
\usepackage{rotating}
\usepackage{footnote}
\usepackage{tikz}
\usetikzlibrary{decorations.shapes}
\usetikzlibrary{plotmarks}
\makesavenoteenv{table}

\newcommand{\blind}{0}

\newtheorem{theorem}{Theorem}
\newtheorem{lemma}[theorem]{Lemma}

\newtheorem{corollary}[theorem]{Corollary}

\definecolor{morange}{rgb}{1, 0.271, 0}
\definecolor{mmagenta}{rgb}{1,0,1}
\definecolor{mpurple}{rgb}{0.63,0.13,0.94}
\definecolor{mpink}{rgb}{1,0.078,0.576}
\definecolor{mForestGreen}{rgb}{0.133,0.545,0.133}

\begin{document}

\def\spacingset#1{\renewcommand{\baselinestretch}%
{#1}\small\normalsize} \spacingset{1}

%%%%%%%%%%%%%%%%%%%%%%%%%%%%%%%%%%%%%%%%%%%%%%%%%%%%%%%%%%%%%%%%%%%%%%%%%%%%%%

\if0\blind
{
  \title{\bf A Goodness-of-Fit Test for Statistical Models}
  \author{Hangjin Jiang \\
    Center for Data Science, Zhejiang University\\}
  \maketitle
} \fi

\if1\blind
{
  \bigskip
  \bigskip
  \bigskip
  \begin{center}
    {\LARGE\bf A Goodness-of-Fit Test for Statistical Models}
\end{center}
  \medskip
} \fi

\bigskip
\begin{abstract}
Statistical modeling plays a fundamental role in understanding the underlying mechanism of massive data (statistical inference) and predicting the future (statistical prediction). Although all models are wrong, researchers try their best to make some of them be useful. The question here is how can we measure the usefulness of a statistical model for the data in hand? This is key to statistical prediction. The important statistical problem of testing whether the observations follow the proposed statistical model has only attracted relatively few attentions. In this paper, we proposed a new framework for this problem through building its connection with two-sample distribution comparison.  The proposed method can be applied to evaluate a wide range of models. Examples are given to show the performance of the proposed method. 
\end{abstract}

\noindent%
\vfill

\newpage
\spacingset{1.45} % DON'T change the spacing!

\section{Introduction}
Statistical prediction and inference is the core of statistics. Statistical prediction is a very important contribution of statistical research to the society, and it depends more strongly than statistical inference on the statistical model proposed by the expert or learned from the past.
However, as George Box said that ``all models are wrong, but some are useful". Here comes a question that  ``To what extent, they (statistical models) are useful?" Or we may ask in another way that ``Is our statistical model built for given dataset acceptable or good enough?" This paper seeks a statistical answer to this question.

Here is a motivating example. Assume that we have observations $\{(x_i,y_i): i=1, 2, \cdots, 200\}$ from model $Y=10(X-0.2)^2+\epsilon$, with $X \sim U(0,1)$ and $\epsilon \sim N(0,1)$.  The scatter plot of these observations is shown in the left panel of Figure~\ref{example}. The simplest way of analyzing these observations is to assume that they follow a linear model, i.e, $Y=\beta_0+\beta_1X+\eta$.  We obtain, by least square method, $\hat{\beta}_0=-1.29 \pm 0.1784$ with $ \text{p-value}=1.09\times 10^{-11}$ and  $\hat{\beta}_1=5.93 \pm 0.3061$ with $\text{p-value}< 2\times 10^{-16}$, and the fitted line is shown as the blue solid line in left panel of Figure~\ref{example}.  

\begin{center}
	\begin{figure}[!htbp]
		\centering
		\includegraphics[scale=0.9]{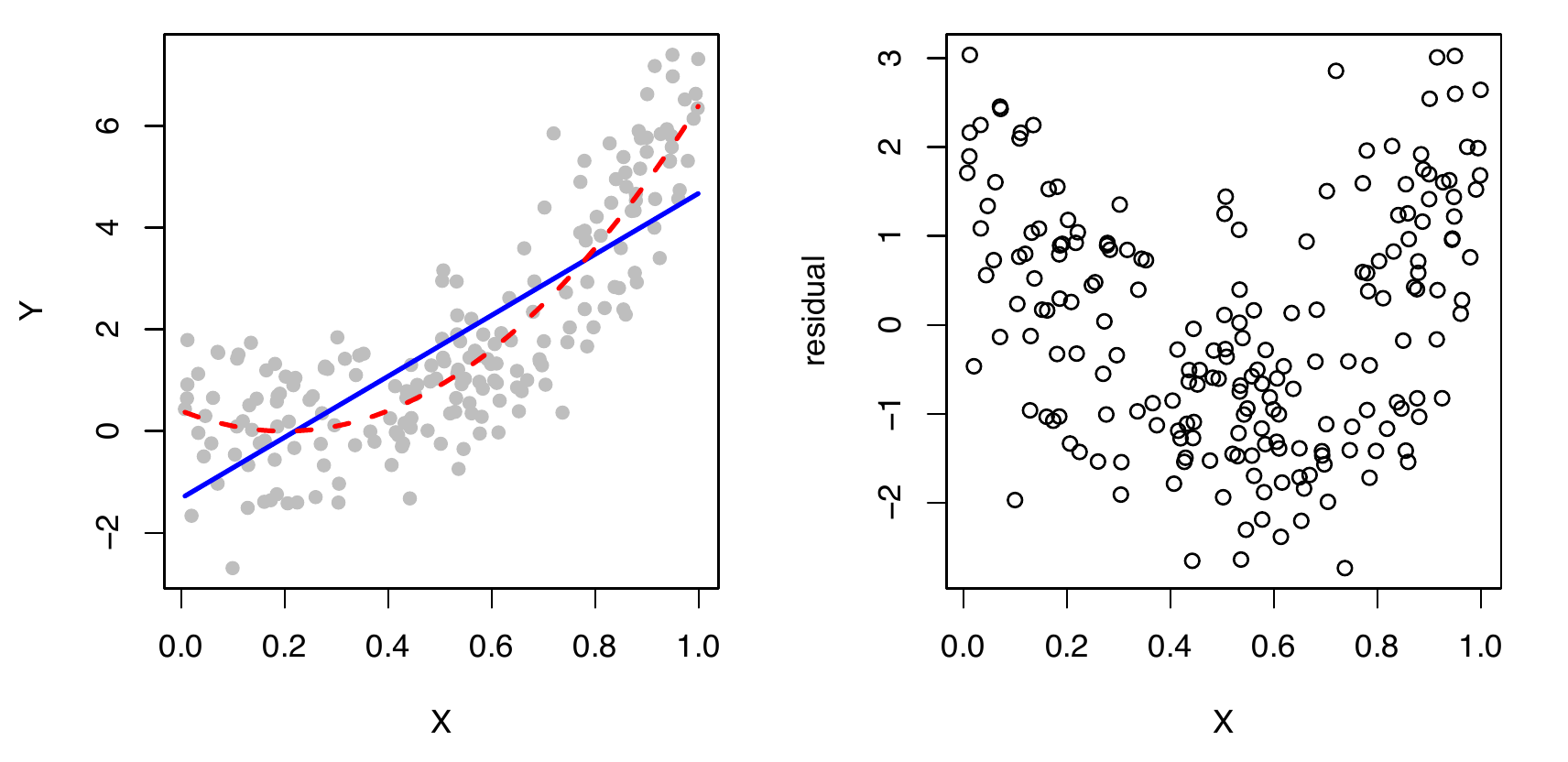}\\
		\caption{A motivating example. Left panel: the gray points are observations from model $Y=10(X-0.2)^2+\epsilon$, with $X \sim U(0,1)$ and $\epsilon \sim N(0,1)$. The blue solid line is obtained by linear regression of these observations, and the red dashed line is the true model.  Right panel is the plot of residuals.  As one can tell from the plot, there is a quadratic trend between the covariate and the residual, but gives no information on to what extent we can believe in the linear model. The p-value from our method for this linear model is 0.004, which indicates the linear model is not good enough.}
		\label{example}
	\end{figure}
\end{center}

The p-value of $\beta$ only tells the response $Y$ strongly depends on $X$ but noting about the question we cared much about that whether the linear model is good enough for these observations? A usual way to answer this question is to fit a more complicated model, for example, a second order model, to the data, and compare it with the linear model.  However, there are so many models and we should try, in principle, each of them to answer this question, which is impossible.  One may resort to the adjusted $R^2$, which is 0.6536 for this case. However, it also tells less information about the goodness-of-fit of the linear model, since there exists two datasets with the same adjusted $R^2$, but one dataset truly follows a linear model and the other one does not (Details are given in Section 3). On the other hand, the adjusted $R^2$ is not applicable to many applications.

Alternatively, by applying our method (see Section 2) to this example, we get a p-value=0.004, which shows a strong evidence against the claim that the linear model is good enough (Details are given in Section 3) and consists with the quadratic trend between the covariate and the residual shown in right panel of Figure~\ref{example}.  The merit of our method is that we do not need to compare the current model with a more complex model to tell whether the current model is correct or good enough, and also we do not need to check the model by looking into the residual plot, which is unavailable for many models, such as logistic model and Poisson model.

The adjusted $R^2$ is a summary statistic of the residuals. Along this direction,  there are many model checking methods proposed for linear regression models defined in equation~(\ref{tmodel}) by analyzing the behavior of residuals based on the assumption of $\varepsilon$, such as $E\varepsilon=0$ or $E(\varepsilon|X)=0$, see for example, \citet{stute1998model,stute1997nonparametric,stute2002model,xia2009model,khmaladze2004martingale, van2008goodness, stute1998bootstrap, zheng1996consistent}.  There are also some other methods proposed for other specified models, for example, \citet{guo2016model} considered model checking for parametric single-index models, \citet{van2008goodness} and \citet{neumeyer2010estimating} considered the case for heteroscedastic regression model, and \citet{zhang2011testing} considered the case for non-stationary time series model. We refer to \citet{gonzalez2013updated} for a review. Note that all of these methods can not be applied to generalized linear models \citep{mccullagh2018generalized},  since there is no residuals for these models. Importantly, they try to answer the question that whether the observations are from the specified model, but we try to answer whether the working model is good enough for the observations in hand.

\cite{fan2001generalized} proposed to use generalized likelihood ratio test to test whether the true model is linear, i.e., testing $H_0:$ $f(\textbf{X})=\beta_0+\beta_1\textbf{X}$ v.s. $H_1:$ $f(\textbf{X})\neq \beta_0+\beta_1\textbf{X}$. Here $H_1$ is usually taken as $f(\textbf{X}) \in S_{k}$ (e.g. $S_k=\{f\in L^2[0,1]: \int_0^1 f^{(k)}(x)^2 dx \leq C\}$), a larger functional space containing $H_0$. Essentially, it is to seek the best approximation of $f(\textbf{X})$ in $S_k$, and compare it with the model under $H_0$. Although it can be done through basis function expansion and work well in some cases, it still has some difficulties in many real applications. Also, there are some works focusing on identifying the structure of the true model $f(\textbf{X})$, say, is it linear or not? which can be achieved by adding a penalty on high order terms, see for example \cite{cao2016sieve}. 

\cite{tsay1992model} proposed a model checking procedure via parametric Bootstrap.  The underlying principle is that if the model is correct or good enough, the Bootstrap 95\% CI for each response $y_i$ should contain the observation $y_i$ with probability close to 95\%.  Finally, a plot with acceptance envelopes is used as a final diagnosis of the model. However, there is no quantitative measurement on the quality of the proposed model. Although we also use the idea of Bootstrap, our method tending to provide a overall evaluation of the working model is totally different from the method in \cite{tsay1992model} that tends to focus on one aspect of the working model.
This Bootstrap idea was also applied to Bayesian model checking called posterior predictive assessment  based on $\chi^2$ discrepancy between the observed response and predicted response \citep{gelman1996posterior}.  Comparing with this method, our method directly focus on the discrepancy between the true model and the working model through comparing the joint distribution between $(\textbf{X}, Y)$ and  $(\textbf{X}, Y^*)$, where $Y^*$ is the predicted response. 
In addition, our method is applicable to both frequentist and Bayesian framework, and computational efficient.

The rest of this paper is organized as follows: In Section 2, we build the connection between measuring the goodness-of-fit of the working model and two-sample distribution comparison.  Furthermore, we proposed a new test statistic, Maximum Adjusted Chi-squared (MAC) statistic, for comparing two unknown distributions based on two independent samples. In Section 3, simulations and real data examples are given to evaluate the performance of our method. Section 4 concludes this paper with some further discussions.

\section{Statistical Model Assessment}
In this section, we will firstly build a connection between statistical model checking and two-sample distribution comparison. After that, we will introduce our new method for the later problem. 
\subsection{Linear  Model}
Firstly, we consider model assessment for the following linear model
\begin{equation}\label{tmodel}
Y=c+f(\textbf{X})+\varepsilon, \text{with } E\varepsilon=0 \text{ and } E\varepsilon^2 < +\infty,
\end{equation}
where $\textbf{X}$ is a random vector in $R^p$ with $p \geq 1$, and $Y$ is a random variable. To make the constant $c$ be estimable, we need the following identifiable assumption.

\textbf{Assumption A.} We assume that $Ef(\textbf{X})=0$, and $c=0$ for simplicity.

The goal of statistical modeling is to find a parametric model $g(\textbf{X}|\theta)$ to better approximate the unknown function $f(
\textbf{X})$ in model~(\ref{tmodel}), where parameter $\theta$ is unknown and will be estimated from observations $\{(\textbf{x}_i, y_i): i=1,2,\cdots, n\}$. If the distribution of $\varepsilon$ is assumed to be normal, which is common in applied statistics, $\theta$ can be estimated through maximizing the likelihood function or Bayesian method.  If one does not put any assumption on the distribution of $\varepsilon$, $\theta$ can be estimated by least square method or estimating equation approach.

The key problem we consider here is to measure the distance between the true model $f(\textbf{X})$ and the working model $g(\textbf{X}|\theta)$.  If the true model is known, the Kullback-Leibler Divergence defined in equation~(\ref{kl}) can be used to measure the goodness of the working model, which also indicates the principle of maximum likelihood.
\begin{equation}\label{kl}
KL(f,g)=\int f \log(f/g) d\mu.
\end{equation}
However, the real situation here is that we have no information about the true model $f$, and what we only have is the observations $\{(\textbf{x}_i, y_i): i=1,2,\cdots, n\}$ from the true model.  Thus, our concern here is to measure the goodness of the working model based on observations from the true model. We summarize this problem in the following.

\textbf{Problem A: Model Assessment.} \textit{Assume that we proposed a working model $g(\textbf{X}|\theta)$ for observations from model~(\ref{tmodel}), to what extent we can believe in the model $g(\textbf{X}|\hat{\theta}_n)$? where $\hat{\theta}_n$ is the estimator of $\theta$. In other words, is there any chance for us to improve our working model?}

To motivate our method, we assume that the proposed working model with identifiable assumption (\textbf{Assumption B}) for observations from model~(\ref{tmodel}) is,
\begin{equation}\label{wmodel}
Y=c+g(\textbf{X}|\theta)+\eta, \text{with } \eta=f(\textbf{X})-g(\textbf{X}|\theta)+\varepsilon.
\end{equation}

\textbf{Assumption B.} $Eg(\textbf{X}|\theta)=0$ for any $\theta$.

\textbf{Remark 1.} \textit{When we do statistical analysis, e.g, statistical inference and statistical prediction, based on working model~(\ref{wmodel}), the underlying assumption we made is that the true model $f(\textbf{X})$ belongs to the parametric model $g(\textbf{X}|\theta)$.}

If the true model $f(\textbf{X})$ belongs to the parametric model $g(\textbf{X}|\theta)$, there exits some $\theta_0$ such that $f(\textbf{X})=g(\textbf{X}|\theta_0)$.  $\eta$ in model~(\ref{wmodel}) can still satisfy $E\eta=0$ under \textbf{Assumption A} and \textbf{B}, and $\theta$ can be estimated by least square method. With this working model, \textbf{Problem A} can be formulated as the following statistical hypothesis testing problem:
\begin{center}
	$H_0$: The true model $f(\textbf{X})$ belongs to the parametric model $g(\textbf{X}|\theta)$; $H_1$: it does not.
\end{center}
Note that to make statistical inference on this testing problem, unlike the generalized likelihood ratio test \citep{fan2001generalized}, we do not need to specify the alternative hypothesis or estimate $f$ in a larger functional space, and thus, our method is very simple to use in real applications.

Denoting the estimates of $\theta$ in model~(\ref{wmodel}) as $\hat{\theta}_n$, 
we get estimates of $\eta$ denoted by $\hat{\eta}=\{\hat{\eta}_i: i=1, 2, \cdots, n \}$ corresponding to $\{(\textbf{x}_i, y_i): i=1,2,\cdots, n\}$ with $\hat{\eta}_i=y_i-g(\textbf{x}_i|\hat{\theta}_n)$.
Let $\eta^*=\{\eta^*_i: i=1, 2, \cdots, n\}$ be an i.i.d bootstrap sample from $\hat{\eta}$ (see Remark 1), then $Y^*=g(\textbf{X}|\hat{\theta}_n)+\eta^*$ is an i.i.d bootstrap sample of $Y$.  With these information in hand, we have the following lemma.

\textbf{Remark 2.}  \textit{The bootstrap samples of $Y$, $Y^*$ can also be obtained through parametric bootstrap with the assumption  $\eta \sim N(0,\sigma_{\eta}^2)$, where $\sigma_{\eta}$ can be estimated from $\hat{\eta}=\{\hat{\eta}_i: i=1, 2, \cdots, n \}$.}

\begin{lemma} Assuming that  $\hat{\theta}_n \rightarrow \theta_0 \text{ a.s.} $ as $n\rightarrow \infty$, we have $P\{(\textbf{X},Y^*)\overset{d}{=}(\textbf{X}, Y)\}=1$, if and only if the true model $f(\textbf{X})$ belongs to the parametric model $g(\textbf{X}|\theta)$, where $X\overset{d}{=} Y$ means that $X$ and $Y$ has the same distribution.
\end{lemma}

\begin{proof}
	`If' part: Since the true model $f(\textbf{X})$ belongs to the parametric model $g(\textbf{X}|\theta)$, there exists a $\theta_0$ such that $f(\textbf{X})=g(\textbf{X}|\theta_0)$. On the other hand, we have $g(\textbf{X}|\hat{\theta}_n)\rightarrow g(\textbf{X}|\theta_0)=f(\textbf{X})$, a.s, as $n\rightarrow \infty$, by $\hat{\theta}_n \rightarrow \theta_0 \text{ a.s} $. Thus, the bootstrap sample $\eta^*$ is the bootstrap sample of $\varepsilon$ in model~(\ref{tmodel}). Finally, the joint density function of $(\textbf{X}, Y^*)$, $p(\textbf{X}, Y^*)=p(Y^*|\textbf{X})p(\textbf{X})=p(Y|\textbf{X})p(\textbf{X})=p(\textbf{X}, Y)$, as $n\rightarrow \infty$. \\
	`only If' part: According to the definition of $\eta^*$, we can write it as $\eta^*=f(\textbf{x}^*)-g(\textbf{x}^*|\hat{\theta}_n)+\varepsilon^*$, where $\textbf{x}^*$ is one of $\textbf{x}_i$ in $\{\textbf{x}_1, \cdots, \textbf{x}_n\}$, and $\varepsilon^*$ the corresponding unknown $\varepsilon_i$.  Given $\textbf{X}$, we have $Y^*=g(\textbf{X}|\hat{\theta}_n)+\eta^*=g(\textbf{X}|\hat{\theta}_n)+f(\textbf{x}^*)-g(\textbf{x}^*|\hat{\theta}_n)+\varepsilon^*$. Comparing with the conditional distribution of $Y|\textbf{X}$, we know that if $f(\textbf{X})=g(\textbf{X}|\hat{\theta}_n)+f(\textbf{x}^*)-g(\textbf{x}^*|\hat{\theta}_n)$, we have $(\textbf{X},Y^*)\overset{d}{=}(\textbf{X}, Y)$. Fixing $\textbf{x}^*$, $c_n=f(\textbf{x}^*)-g(\textbf{x}^*|\hat{\theta}_n)$ is a constant approaching $c=f(\textbf{x}^*)-g(\textbf{x}^*|\theta_0)$, as $n\rightarrow \infty$, and $f(\textbf{X})=g(\textbf{X}|\hat{\theta}_n)+c_n \rightarrow g(\textbf{X}|\theta_0)+c$ should hold for any $\textbf{X}$. Under \textbf{Assumption A} and \textbf{B}, this is only the case when $f(\textbf{X})=g(\textbf{X}|\theta_0)$.
\end{proof}
\textbf{Lemma 1} bridges the original problem of measuring the goodness of the working model to the problem of two-sample distribution comparison, i.e., measuring the distance between the distribution of observations $\{(\textbf{x}_i, y_i): i=1,2,\cdots, n\}$ from the true model and the distribution of bootstrap samples  $\{(\textbf{x}_i, y_i^*): i=1,2,\cdots, n\}$ from the working model.  In Section 2.3, we will introduce a new method for comparing two distributions of two independent samples.
\subsection{Generalized Linear  Model}
In this section, we consider model assessment for the generalized linear model:
\begin{equation}\label{tgmodel}
Y \sim f(y; \beta, \phi)=\exp\{\frac{y\beta-b(\beta)}{\phi+c(y,\phi)}\}, \ \ h(\mu)=h(EY)=c+f(\textbf{X}),
\end{equation}
where the response $Y$ has a distribution belonging to the exponential family and the covariates $\textbf{X}$ are linked to the response through a link function $h$.

Under this setting, we want to assess the following working model:
\begin{equation}\label{wgmodel}
Y \sim f(y; \beta, \phi)=\exp\{\frac{y\beta-b(\beta)}{\phi+c(y,\phi)}\}, \ \ h(\mu)=h(EY)=c+g(\textbf{X}|\theta).
\end{equation}
Similar to the linear model, we have following result.  Its proof is simple, and omitted.
\begin{lemma} Assuming that  $\hat{\theta}_n \rightarrow \theta_0 \text{ a.s.} $ as $n\rightarrow \infty$, we have $P\{(\textbf{X},Y^*)\overset{d}{=}(\textbf{X}, Y)\}=1$, if and only if the true model $f(\textbf{X})$ belongs to the parametric model $g(\textbf{X}|\theta)$, where $X\overset{d}{=} Y$ means that $X$ and $Y$ has the same distribution, where $Y^*$ is generated from model~(\ref{wgmodel}) with $\theta$ replaced by its estimator $\hat{\theta}_n$.
\end{lemma}
\subsection{Two-sample Distribution Comparison}
According to Lemma 1 and Lemma 2, in this section, we will consider the following two-sample distribution comparison problem: 

\textbf{Problem B: Distribution Comparison.}  \textit{Given observations of $(\textbf{X},\textbf{Y}) \in R^p\times R^q$, $\{(\textbf{x}_i, \textbf{y}_i): i=1, 2, \cdots, n\}$ and  observations of $(\textbf{X},\textbf{Y}^*)\in R^p\times R^q$, $\{(\textbf{x}_i, \textbf{y}_i^*): i=1, 2, \cdots, n\}$, where $p \geq 1$ and $q \geq 1$ are two positive integers. We want to test null hypothesis:  $(\textbf{X},\textbf{Y}^*)\overset{d}{=}(\textbf{X}, \textbf{Y})$}

There are many methods proposed for this problem, see \citet{thas2010comparing} for a review and see \citet{zhou2017two, chen2017new} and \citet{kim2018robust} for recent developments on this direction.
\citet{jiang2020} introduced a new class of statistics, Maximum Adjusted Chi-squared (MAC) test statistics, for this problem, which is promising as seen from their results. In the following, we give a brief introduction on MAC.

To define our new statistic, some notations are needed. Let $\textbf{L}_i=(\textbf{w}_i, \textbf{v}_i), i=1, 2, \cdots, k$, be $k$ points in  
$R^p\times R^q$, where $\textbf{w}_i \in R^p$ and $\textbf{v}_i \in R^q$. 
Define  
$A_{\textbf{L}_i,\textbf{L}_j}=\{(\textbf{x},\textbf{y}) \in R^p\times R^q:
d(\textbf{x},\textbf{w}_{i})\leq d(\textbf{w}_{i},\textbf{w}_{j})\}$, and
$B_{\textbf{L}_i,\textbf{L}_j}=\{(\textbf{x},\textbf{y}) \in R^p\times R^q:
d(\textbf{y}, \textbf{v}_{i})\leq d(\textbf{v}_{i},\textbf{v}_{j})\}$.
Furthermore, we define $A_{11}=A_{\textbf{L}_i,\textbf{L}_j}\bigcap
B_{\textbf{L}_i,\textbf{L}_j}$,
$A_{12}=A_{\textbf{L}_i,\textbf{L}_j}^c\bigcap
B_{\textbf{L}_i,\textbf{L}_j}$,
$A_{21}=A_{\textbf{L}_i,\textbf{L}_j}\bigcap
B_{\textbf{L}_i,\textbf{L}_j}^c$,
$A_{22}=A_{\textbf{L}_i,\textbf{L}_j}^c\bigcap
B_{\textbf{L}_i,\textbf{L}_j}^c$. Let $P_{ij}=\sum_{k=1}^n
I((\textbf{x}_k, \textbf{y}_k) \in A_{ij})$, $Q_{ij}=\sum_{k=1}^n I((\textbf{x}_k, \textbf{y}_k^*)
\in A_{ij})$, and $R_{ij}=P_{ij}+Q_{ij}$, for $i=1, 2; j=1, 2$,
the local test statistic at
$(\textbf{L}_i,\textbf{L}_j)$ is defined as
\begin{equation}\label{local}
S(\textbf{L}_i,\textbf{L}_j)=\sum_{k=1}^2\sum_{l=1}^2 \frac{(P_{kl}-Q_{kl})^2}{R_{kl}}.
\end{equation}
If $R_{kl}=0$, we define the corresponding term in the
above formula as 0. It can be shown that $S(\textbf{L}_i,\textbf{L}_j) \rightarrow \chi_3^2$, as $n\rightarrow \infty$ (see Lemma~\ref{lA1} in Appendix A.2), and thus $S(\textbf{L}_i,\textbf{L}_j)$ is called adjusted chi-squared statistic. Note that the definition of the local statistic is different from that given in \citet{jiang2020}.

Define $\textbf{D}=(\textbf{X}, \textbf{Y})$ and $\textbf{D}^*=(\textbf{X}, \textbf{Y}^*)$, the Maximum Adjusted Chi-squared (MAC) test statistic is defined as the maximum of all local statistics:
\begin{equation}
\text{MAC}(\textbf{D}, \textbf{D}^*)=\max_{1\leq i< j
	\leq k} S(\textbf{L}_i,\textbf{L}_j).
\end{equation}
We reject $H_0:  \textbf{D} \overset{d}{=} \textbf{D}^*$ if $\text{MAC}(\textbf{D}, \textbf{D}^*)>c_0$, where $c_0$
is a positive threshold.  More discussions on the definition of MAC are given in Appendix A.1.

\textbf{Remark 3.} \textit{The principle underlying MAC is as follows.  For any partition of the space $\Omega=A_{11}\bigcup A_{12} \bigcup A_{21} \bigcup A_{22}$, we can construct a $\chi^2$ test statistic $S$, and $S\rightarrow \chi_3^2$ if $H_0:  \textbf{D} \overset{d}{=} \textbf{D}^*$ is true.  In other words, the local statistic $S(\textbf{L}_i,\textbf{L}_j)$ is defined on a particular partition of the space $\Omega$. In principle, we should check any kind of partition of the space to see whether there is any difference of these two distributions on this partition.  In real applications, due to limited observations from the underlying distribution, we only need to check a part of them. This idea is similar to that behind HHG \citep{heller2012consistent}.
}

To establish the consistency of MAC statistic, we define the population version of the local test statistic at $(\textbf{L}_i,\textbf{L}_j)$ as
\begin{equation}\label{plocal}
S^p(\textbf{L}_i,\textbf{L}_j)=n\sum_{k=1}^2\sum_{l=1}^2 \frac{(p_{kl}-q_{kl})^2}{r_{kl}}.
\end{equation}
where $p_{kl}=P_{\textbf{D}}(A_{kl})$, $q_{kl}=P_{\textbf{D}^*}(A_{kl})$, and $r_{kl}=p_{kl}+q_{kl}$. Furthermore, the population version of MAC statistic is given by
\begin{equation}
\text{MAC}^p(\textbf{D}, \textbf{D}^*)=\max_{1\leq i < j
	\leq k} S^p(\textbf{L}_i,\textbf{L}_j).
\end{equation}
The consistency of MAC is established in Theorem \ref{consist}, whose proof is given in Appendix A.2.

\begin{theorem}\label{consist} Following above notations and assuming that $r_{ij}\geq \eta>0 $ for any $i$ and $j$, then for any $\epsilon>0$, let $\delta=\frac{4n\epsilon}{\eta}(\frac{1}{\eta}+2)$, we have
	$$
	P(|{\rm MAC}(\textbf{D}, \textbf{D}^*)-{\rm MAC}^p(\textbf{D}, \textbf{D}^*)|>\delta)< 4k^2e^{-n\epsilon^2/8}.
	$$
	Furthermore, we have: (1) under $H_0$, $P({\rm MAC}(\textbf{D}, \textbf{D}^*)\leq O(\sqrt{n\log k}))\rightarrow 1$, as $n\rightarrow \infty$; (2) under $H_1$, $P({\rm MAC}(\textbf{D}, \textbf{D}^*)> O(n))\rightarrow 1$, as $n\rightarrow \infty$.
\end{theorem}

\textbf{Remark 5.} \textit{As stated in the theorem, the upper bound of MAC under the null hypothesis is of order $\sqrt{n\log(k)}$, but its lower bound under the alternative hypothesis is of order $n$. This implies the consistency of MAC.}

\textbf{Remark 6.} \textit{The upper bound of MAC under the null hypothesis does not depend on where we define the local statistics, i.e, the choice of $\textbf{L}_i$, but only depends on the number of these locations.  This property is key to our algorithm of model checking.
}

\subsection{The Algorithm}
In this section, we introduce our algorithm for statistical model assessment with discussions on other methods.  The algorithm goes in two steps: (1) estimate the discrepancy between the true model and the working model by statistic defined in equation (2.7) and (2) get the p-value of this discrepancy.  In Step 1, the discrepancy is estimated by comparing the distribution of observations from the true model and that of predictive observations obtained from the working model via Bootstrap.  
In Step 2, the p-value is obtained through simulation. The algorithm is summarized in \textbf{Algorithm 1}, which is called predictive assessment of statistical models,  since it is on the basis of predictive observations of the response. 

The idea underlying our algorithm is similar to that in \cite{tsay1992model} and \cite{gelman1996posterior}.  \cite{tsay1992model} estimates the $(1-\alpha)$ confidence interval for each observation $y_i (i=1, 2, \cdots, n)$ based on the working model. If the working model is good enough, this the interval shall cover the observation $y_i$ with probability close to 95\%. Based on this principle, the working model is checked using a plot called acceptance envelopes.  However, there is no overall assessment of the working model.  The Bayesian predictive assessment method in \cite{gelman1996posterior} assess the working model through a predefined discrepancy statistic, for example, $\chi^2$ discrepancy, which only focuses on one aspect of the discrepancy between the true model and the working model.  In addition,  for the Bayesian predictive assessment method, if we have a wrong model for the error term $\varepsilon$ in model~(\ref{tmodel}), we also obtain a small p-value, which is undesired.
Comparing with these two methods, the key ingredient in our algorithm is that we use MAC statistic, $T_B$ in \textbf{Algorithm 1}, to measure the discrepancy between the true model and the working model through comparing the distribution of observations from the true model and that of predictive observations from the working model.  This ingredient allow us to give a full characterization of the discrepancy between the true model $f(\textbf{X})$ and the working model $g(\textbf{X}|\theta)$. In other words, $T_B$ can be regarded as the distance between the true model and the working model.

As one can see, the local statistics $S(\textbf{L}_i,\textbf{L}_j)$ are dependent, and some of them are strongly correlated. This phenomenon occurs also in U-statistics. The reduced U-statistics (e.g., \cite{blom1976some} and \cite{brown1978reduced}), counting only a part of all terms, is proposed to improve the computational efficiency when the sample size is very large.  In our case, we should, in principle, construct local statistics based on all kinds of different partitions of the space (see Remark 3), and take the maximum of them as our test statistic. However, it is enough to consider the partitions based on observations of random vectors.  Due to the dependence of local statistics constructed from these observations, we can only focus on a part of them to reduce the computational cost without much loss of power when the sample size is very large. There are many kinds of strategies to determine where to define the local statistics. Here is two examples: (1) we can randomly select $k$ observations from observations of $(\textbf{X}, \textbf{Y})$ and $(\textbf{X}, \textbf{Y}^*)$, and define local statistics on these selected observations; (2) one may cluster the observations of $(\textbf{X}, \textbf{Y})$ and $(\textbf{X}, \textbf{Y}^*)$ into $k$ clusters and define local statistics on the center of each cluster.

\begin{algorithm}[tb]
	\caption{Goodness-of-Fit Test for Statistical Models}
	\label{alg:algorithm}
	\textbf{Input}: observations: $\{(\textbf{x}_i,y_i): i=1, 2, \cdots, n\}$; working model: $g(\textbf{X}|\hat{\theta}_n)$;
	
	\textbf{Output}: p-value of checking whether observations come from the working model
	
	\begin{algorithmic}[1] %[1] enables line numbers
		\FOR{$b=1, \cdots, B$}
		\FOR{$i=1, 2, \cdots, n$}
		\STATE{generate Bootstrap sample of $y_i$, $y_i^{*b}$, by Bootstrapping the residuals or parametric  Bootstrap}
		\ENDFOR
		\STATE{Compute the test statistic defined in equation (2.7), $T_b$, between observations $\{(\textbf{x}_i,y_i): i=1, 2, \cdots, n\}$ and Bootstrap sample $\{(\textbf{x}_i,y_i^{*b}): i=1, 2, \cdots, n\}$}
		\ENDFOR
		\STATE{Estimate the the discrepancy between true model and working model as $T_B=\frac{1}{B}\sum_{b=1}^B T_b$ (taking average to reduce the randomness of $T_b$ from randomness of $y^{*b}$).}
		\STATE{Compute the p-value of $T_B$ based on the empirical null distribution of MAC obtained through simulation.}
	\end{algorithmic}
\end{algorithm}

\section{Numerical Results}
In this section, we show the performance of the proposed method through some examples.  Firstly, four simulation examples are used to show the power of MAC on two-sample distribution comparison. Secondly, five examples are used to illustrate the power of \textbf{Algorithm 1} on model checking.
\subsection{Power of MAC}
In this section, simulation studies are presented to explore the power of MAC developed for testing the null hypothesis $(\textbf{X},\textbf{Y}^*)\overset{d}{=}(\textbf{X}, \textbf{Y})$.  In each of the following examples, we set $\textbf{X} \sim N(0, I_d)$, and set the sample size of the observations $n=\{100, 200, 300\}$.  The power is estimated from 1000 independent simulations at significant level 5\%.

\textbf{Example 1: One-dimensional case.} $Y^*|X \sim N(X, 1)$ and $Y|X \sim N(cX, 1)$, with $c$ varying from 1 to 5.

\textbf{Example 2: Two-dimensional case.} $Y^*|X \sim N(X_1+X_2, 1)$ and $Y|X \sim N(X_1+cX_2, 1)$ with $c$ varying from 1 to 5 and $X=(X_1, X_2)$.

\begin{center}
	\begin{figure}[!htbp]
		\centering
		\includegraphics[scale=0.8]{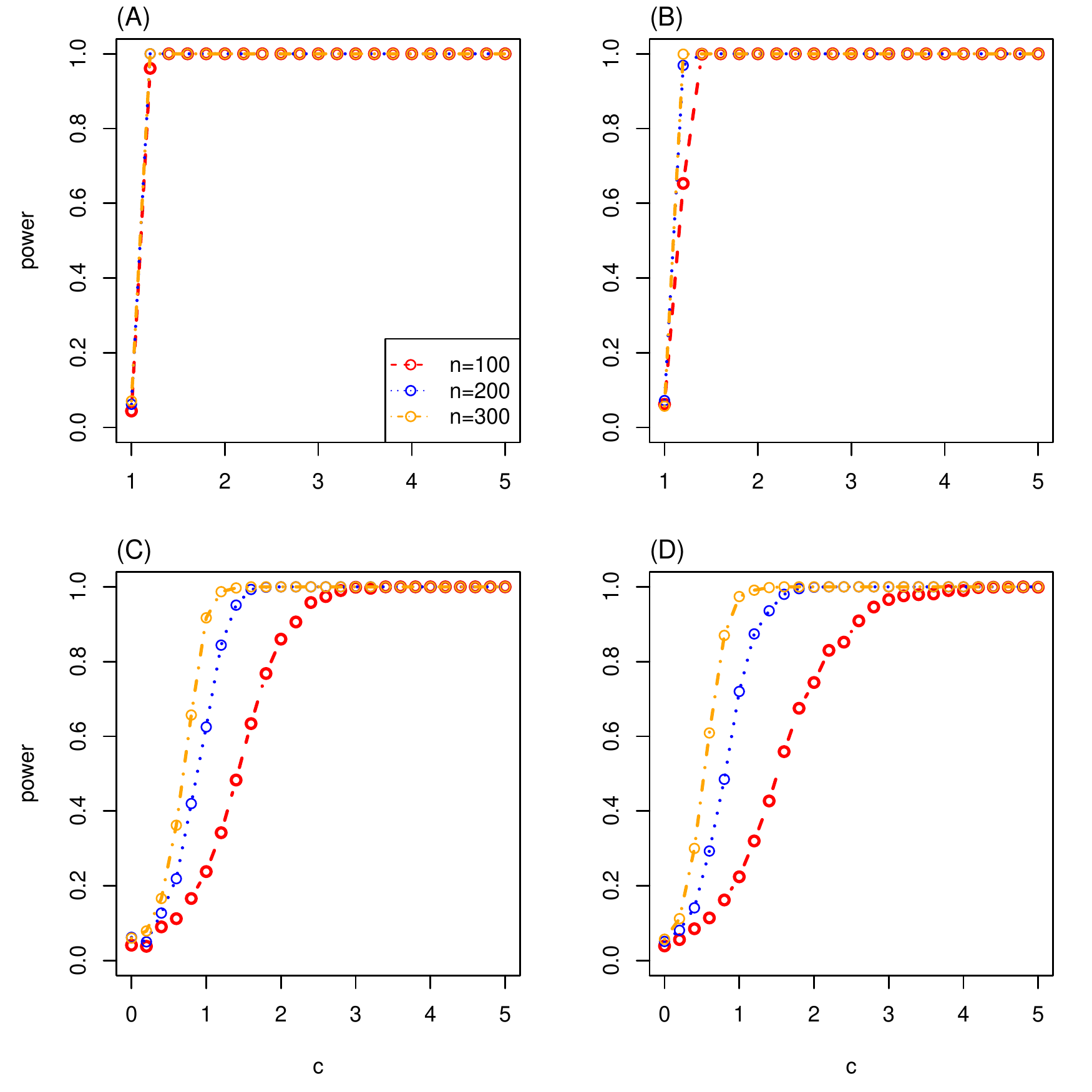}\\
		\caption{Power of MAC for example 1-4 under different settings of $c$: (A) Example 1; (B) Example 2; (C) Example 3; (D) Example 4.}
		\label{simu}
	\end{figure}
\end{center}

\textbf{Example 3: Two-dimensional case with interaction.} $Y^*|X \sim N(X_1+X_2, 1)$ and $Y|X \sim N(X_1+X_2+cX_1X_2, 1)$ with $c$ varying from 0 to 5 and $X=(X_1, X_2)$.

\textbf{Example 4: Five-dimensional case with interaction.} $Y^*|X \sim N(X_1+X_2+X_3+X_4+X_5, 1)$ and $Y|X \sim N(X_1+X_2+X_3+X_4+X_5+cX_1X_2X_3X_4X_5, 1)$ with $c$ varying from 0 to 5 and $X=(X_1, X_2, \cdots, X_5)$.

The power of MAC for these four examples under different settings of $c$ is presented in Figure~\ref{simu}.  As shown by the results, MAC controls the type-I error around the targeted significant level 5\%, and always perform well when the sample size is large enough (consistent with our theoretical result, see Theorem 3). In addition, comparing with the lower dimensional case with the same sample size, MAC tends to be less sensitive to the difference between two distributions, when the dimension of $X$ goes higher. In other words, to attain the same power, MAC needs more samples for higher dimensional case than lower dimensional case.  In principle, MAC can be applied to \textbf{Problem B} under any setting of $p$ and $q$. The only requirement of good performance of MAC for the higher dimensional case is a larger sample size, as indicated by our simulations.

\subsection{Examples of Model Assessment}

\begin{center}
	\begin{figure}[!htbp]
		\centering
		\includegraphics[scale=0.55]{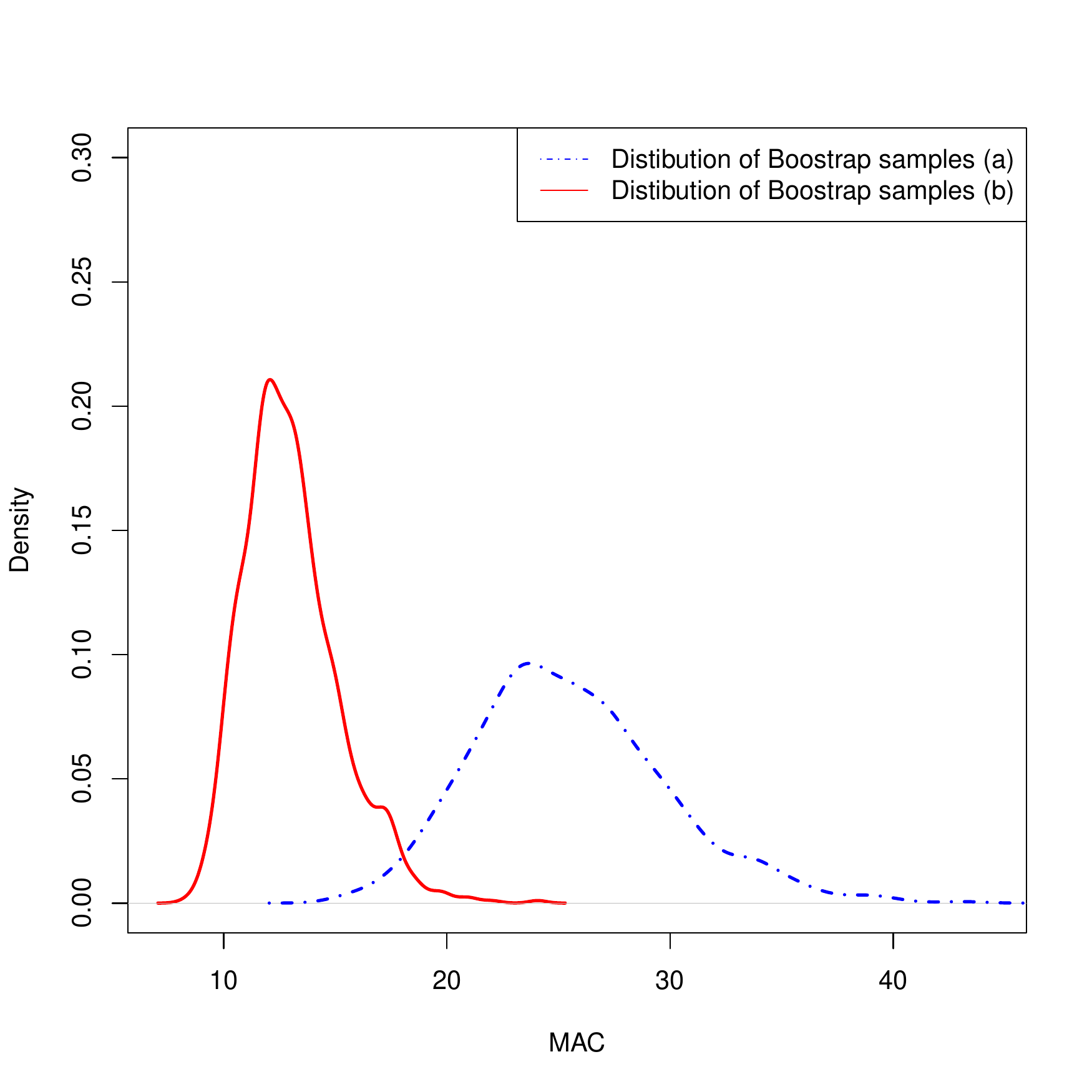}\\
		\caption{Distribution of MAC statistic in Example 5. The red solid line is the distribution of Bootstrap samples of MAC from linear model for case (b), and the blue dash-dotted line is the distribution of Bootstrap samples of MAC from linear model for case (a). }
		\label{example4}
	\end{figure}
\end{center}

In this section, we show numerical results on  five examples of statistical model checking based on \textbf{Algorithm 1}. Two of them are based on synthesized dataset, and the other three are real examples.  In all these examples, we take $B=1000$, and the corresponding null distribution of MAC is obtained from 1000 independent simulations.

\textbf{Example 5:  Simple linear regression model.}
In this example, we generate 200 observations from model (a) $Y=10(X-0.2)^2+\epsilon$ and (b) $Y=4.7X+\epsilon$, where $X \sim U(0,1)$ and $\epsilon \sim N(0,1)$. Firstly, we focus on model (a), and the scatter plot of 200 observations from this model is given in Figure 1. With this dataset in hand, one may want to try a linear model to fit these observations, and get the blue solid line in Figure 1, i.e., $Y=-1.29+5.93X$. All of these coefficients are significant different from 0, and the adjusted $R^2$ is 0.654. Classical method for model checking is to check the behavior of residuals.  For this example, it is easy to see that there is some pattern in the residual plot (right panel in Figure 1), and confirm that linear model is not good enough. But, we do not have any quantification on this under-fitting.  The p-value from \textbf{Algorithm 1} is 0.004, which also indicates the underfitting of linear model to the data.  To show that the adjusted $R^2$ does not tell any information about the quality of the model,  we focus on model (b).  Based on observations from model (b), we obtain adjusted $R^2$  as 0.654. Apply our method to these observations, we get a p-value 0.486, which indicates that  linear model is acceptable/good enough to fit these observations. Note that the adjusted $R^2$ for linear model based on observations from model (a) and model (b) are the same, however, the linear model is under-fitting to the observations from model (a), and it is good enough for observations from model (b). These facts tell us that the adjusted $R^2$ does not tell much about the discrepancy between the true model and the working model.  

Furthermore, to better understand \textbf{Algorithm 1} and the behavior of MAC statistic, we showed in Figure~\ref{example4} the distribution of Bootstrap samples of MAC ($T_b$'s in \textbf{Algorithm 1}) for linear model under case (a) and case (b). Comparing with these two distributions, we see that the distribution of Bootstrap samples of MAC shifts towards left when the working model is not good enough, which is consistent with the claim made in Theorem 3.

\textbf{Example 6:  Simple logistical regression.} In this example, we generate 200 observations from model $\log(P(Y=1)/P(Y=0))=2X+X^2$ with $X\sim N(0,1)$. Now, we fit the data using model $\log(P(Y=1)/P(Y=0))=\beta_0+\beta_1X$, and get the estimate:  $\hat{\beta}_0=0.5939\pm0.033$ with $\text{p-value} < 2\times 10^{-16}$ and $\hat{\beta}_1=0.197 \pm 0.034$ with $\text{p-value} = 2.66\times 10^{-8}$.  There is no residuals for this example, and thus no adjusted $R^2$.  Applying \textbf{Algorithm 1} to this dataset, we have p-value=0.002, which shows the underfitting of linear logistical model to the data.

\textbf{Example 7:  Auto MPG data (Multivariate regression model).} In this example, we revisited the linear model for Auto MPG dataset  proposed by \citet{quinlan1993combining}. The dataset is available on the Machine Learning Repository at the University of California at Irvine (\url{http://archive.ics.uci.edu/ml/datasets/Auto+MPG}). In this dataset, the response is milers per gallon, and there are seven predictors: the number of cylinders, engine displacement, horsepower, vehicle weight, time to accelerate from 0 to 60 m.p.h, model year, and origin of the car. As the origin of the car contains more than two categories, we define two indicator variables to denote whether it comes from America/Europe. After removing items with missing values, there are 392 samples left. Now, we apply \textbf{Algorithm 1} to evaluate the linear model for this dataset, and the results are shown in Figure~\ref{autompg}.  In the left panel of the figure, the red line shows the distribution of MAC statistic under the null and the blue dashed line shows the distribution of Bootstrap samples of MAC ($T_b$'s in \textbf{Algorithm 1}).  There is almost no overlap between these two distributions, which means the working model is far away from the true model (the p-value from \textbf{Algorithm 1}  is 0). We also show in the right panel of Figure~\ref{autompg} the residual plot from linear model. It seems there is some pattern in the residual but we do not know how many evidence in there against the linear model.

\begin{center}
	\begin{figure}[!htbp]
		\centering
		\includegraphics[scale=0.55]{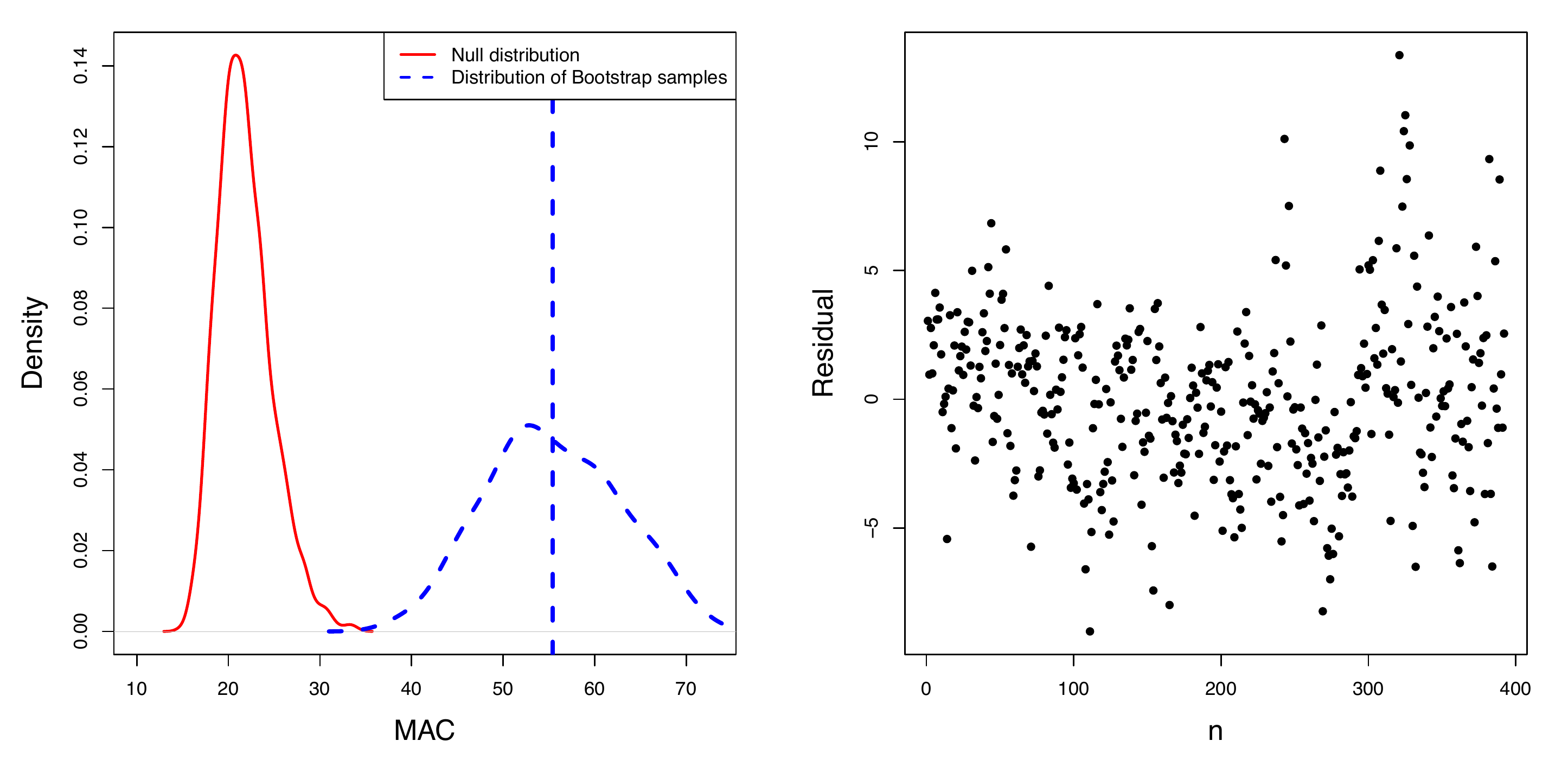}\\
		\caption{Results for Example 7. Left Panel: The empirical null distribution of MAC and the distributions of bootstrap samples of MAC based on linear model. The dashed vertical line is the mean of Bootstrap samples of MAC. Right Panel: the residual plot from linear model.}
		\label{autompg}
	\end{figure}
\end{center}

\textbf{Example 8: Patent data (Poisson Model).} We study the relationship of the number of patent applications to research and development (R\&D) spending and sales using data from 1976 for 70 pharmaceutical and biomedical companies from \citet{hall1988r}. The same data set have been analyzed in \cite{czado2009predictive}, where the authors compared the Poisson model with negative binomial model, and concluded that negative binomial model is a little better than Poisson model. Now, we apply \textbf{Algorithm 1} to evaluate these two models, and get $T_B=13.02$ for Poisson model and $T_B=11.78$ for negative binomial model, which indicates the negative binomial model is a little better. This is consistent with the conclusion in \cite{czado2009predictive}.
Furthermore, we get p-value$>0.1$ for these two models and show that these two models are applicable to this dataset.

\textbf{Example 9: Annual Wolf sunspot data (Time series model).}  In this example, we consider the annual Wolf sunspot numbers from 1700 to 1980, which is available in R dataset by calling data(``sunspots") in R. 
There are three models proposed for this dataset, the AR(9) model, TAR model and AR(2) log-normal model, as shown in the following. These models were analyzed before by \citet{tsay1992model}.

\textbf{AR(9) model}
\begin{eqnarray*}
	z_t&=&6.96+1.21z_{t-1}-0.45z_{t-2}-0.17z_{t-3}+0.20z_{t-4}-0.13z_{t-5} \\
	&+& 0.03z_{t-5}+0.01z_{t-7}-0.03z_{t-8}+0.21z_{t-9}+a_t,
\end{eqnarray*}
where $a_t \sim N(0, 221.24)$.

\textbf{TAR model}
\begin{center}
	$$ z_t= \left\{\begin{array}{ll}
	10.88+1.869z_{t-1}-1.556z_{t-2}+0.086z_{t-3}+0.326z_{t-4}+a_{1t}&  {\rm if} ~~~ z_{t-3}\leq 32.3 \\
	\\
	8.726+0.679z_{t-1}+0.064z_{t-2}-0.217z_{t-3}+0.044z_{t-4}-0.118z_{t-5}\\-0.005z_{t-6}+0.192z_{t-7}-0.285z_{t-8}+0.242z_{t-9}-0.123z_{t-10}\\+0.246z_{t-11}+2a_{2t}&  {\rm if} ~~~ z_{t-3}> 32.3 \\
	\end{array}
	\right. $$
\end{center}
where $a_{1t} \sim N(0, 275.7)$ and $a_{2t} \sim N(0, 82.9)$.

\textbf{AR(2) log-normal model}
\begin{equation*}
z_t=1.6759z_{t-1}-0.7840z_{t-2}+b_t,
\end{equation*}
where $b_t \sim \text{log-normal}(13.88, 153.39)$.

Now, we apply \textbf{Algorithm 1} to evaluate these three models, and we obtain $T_B=94.04$ for AR(9) model, $T_B=92.51$ for TAR model and $T_B=573.0$ for AR(2) log-normal model.  The empirical null distribution of MAC and the distributions of bootstrap samples of MAC from AR(9) and TAR model are shown in Figure~\ref{time}.  According to these results, we concluded that (1) AR(2) log-normal model is very questionable, which is consistent with the finding in  \citet{tsay1992model}; (2) AR(9) and TAR model show similar overall performance, but still have room for improvement, although \citet{tsay1992model} showed that both of them are reasonable in terms of spectral density (second-order property). 

\begin{center}
	\begin{figure}[!htbp]
		\centering
		\includegraphics[scale=0.6]{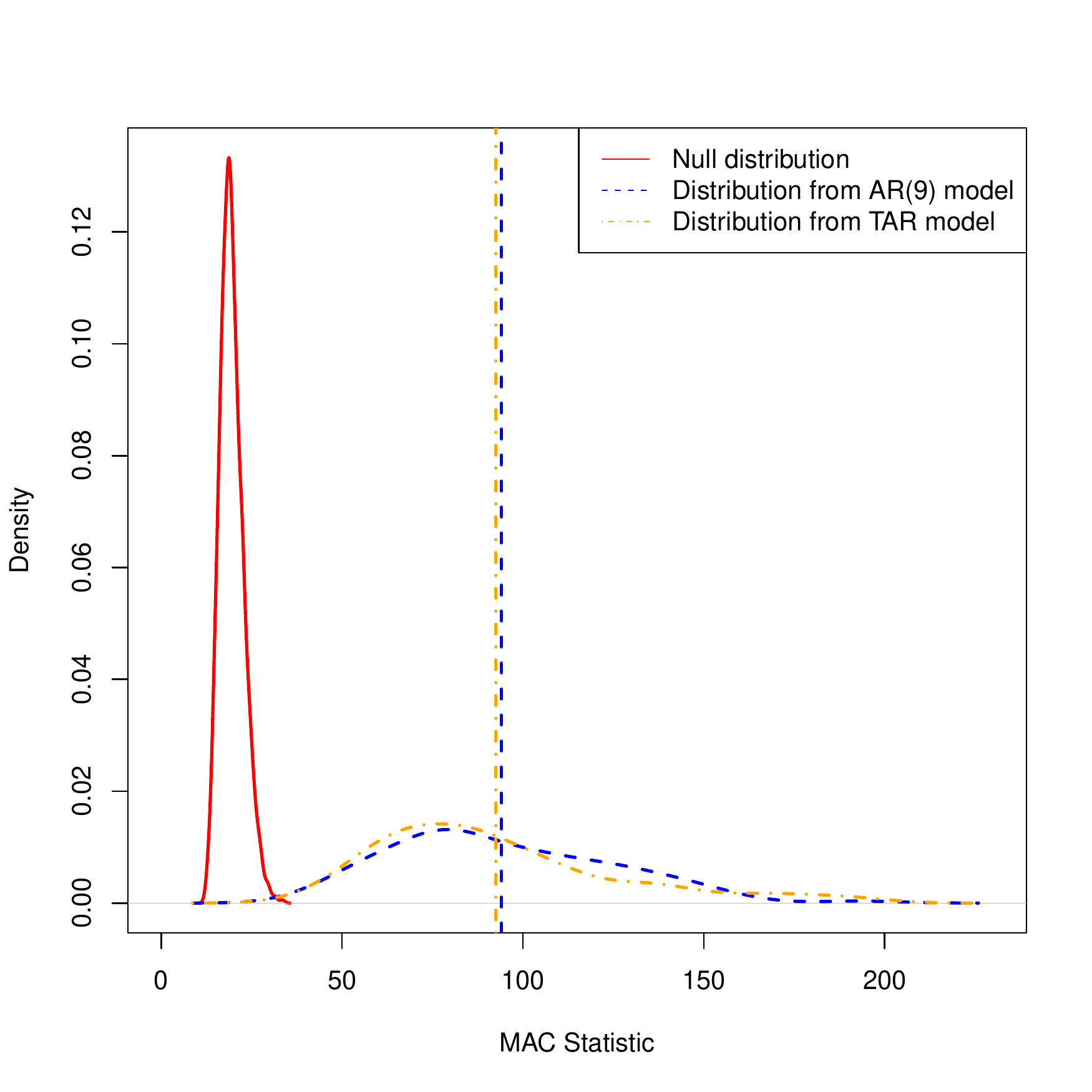}\\
		\caption{Results for Example 9. The empirical null distribution of MAC and the distributions of bootstrap samples of MAC from AR(9) and TAR model. The dashed vertical line is the mean of bootstrap samples of MAC.}
		\label{time}
	\end{figure}
\end{center}

\section{Discussion}
In this paper, we try to seek a answer to the question that how can we believe in our statistical models through the bridge between the problem of measuring the distance between the true model and the working model and two-sample distribution comparison. To the best of our knowledge, it is the first time to notice such kind of connection between these two fundamental statistical problems. We believe that this new connection will be very useful to solve problem in many application areas. Numerical results from both simulation studies and real data analysis showed the advantages of our method. 

We directly target on measuring the distance between the true model and the working model, however, other methods focused on a special aspect of the model through a predefined statistic. In other words, if one only cares about whether some aspects of the model meet the requirement, methods proposed in, for example, \citet{tsay1992model} and \citet{gelman1996posterior} are more appropriate. If one cares about the prediction performance of the model, our method is preferred. 

There are also some other interesting problems related to the framework proposed in this paper. For example, we only show that the MAC statistic defined in equation (2.7) has a upper bound under the null hypothesis and a lower bound under the alternative.  It is very interesting to investigate the limit distribution of MAC.  Furthermore, a detailed mathematical relationship is in need to better understand the connection between two-sample distribution comparison and measuring the distance between the true model and the working model. These interesting topics are left as future works.
\section*{Acknowledgment}
This work is partially supported by .....
\bibliographystyle{plainnat}
\bibliography{ref}

\begin{thebibliography}{28}
\providecommand{\natexlab}[1]{#1}
\providecommand{\url}[1]{\texttt{#1}}
\expandafter\ifx\csname urlstyle\endcsname\relax
  \providecommand{\doi}[1]{doi: #1}\else
  \providecommand{\doi}{doi: \begingroup \urlstyle{rm}\Url}\fi

\bibitem[Blom(1976)]{blom1976some}
Gunnar Blom.
\newblock Some properties of incomplete u-statistics.
\newblock \emph{Biometrika}, 63\penalty0 (3):\penalty0 573--580, 1976.

\bibitem[Brown and Kildea(1978)]{brown1978reduced}
BM~Brown and DG~Kildea.
\newblock Reduced u-statistics and the hodges-lehmann estimator.
\newblock \emph{The Annals of Statistics}, pages 828--835, 1978.

\bibitem[Cao et~al.(2016)Cao, Huang, Liu, and Zhao]{cao2016sieve}
Yongxiu Cao, Jian Huang, Yanyan Liu, and Xingqiu Zhao.
\newblock Sieve estimation of cox models with latent structures.
\newblock \emph{Biometrics}, 72\penalty0 (4):\penalty0 1086--1097, 2016.

\bibitem[Chen and Friedman(2017)]{chen2017new}
Hao Chen and Jerome~H Friedman.
\newblock A new graph-based two-sample test for multivariate and object data.
\newblock \emph{Journal of the American Statistical Association}, 112\penalty0
  (517):\penalty0 397--409, 2017.

\bibitem[Czado et~al.(2009)Czado, Gneiting, and Held]{czado2009predictive}
Claudia Czado, Tilmann Gneiting, and Leonhard Held.
\newblock Predictive model assessment for count data.
\newblock \emph{Biometrics}, 65\penalty0 (4):\penalty0 1254--1261, 2009.

\bibitem[Fan et~al.(2001)Fan, Zhang, Zhang, et~al.]{fan2001generalized}
Jianqing Fan, Chunming Zhang, Jian Zhang, et~al.
\newblock Generalized likelihood ratio statistics and wilks phenomenon.
\newblock \emph{The Annals of Statistics}, 29\penalty0 (1):\penalty0 153--193,
  2001.

\bibitem[Gelman et~al.(1996)Gelman, Meng, and Stern]{gelman1996posterior}
Andrew Gelman, Xiao-Li Meng, and Hal Stern.
\newblock Posterior predictive assessment of model fitness via realized
  discrepancies.
\newblock \emph{Statistica sinica}, pages 733--760, 1996.

\bibitem[Gonz{\'a}lez-Manteiga and Crujeiras(2013)]{gonzalez2013updated}
Wenceslao Gonz{\'a}lez-Manteiga and Rosa~M Crujeiras.
\newblock An updated review of goodness-of-fit tests for regression models.
\newblock \emph{Test}, 22\penalty0 (3):\penalty0 361--411, 2013.

\bibitem[Guo et~al.(2016)Guo, Wang, and Zhu]{guo2016model}
Xu~Guo, Tao Wang, and Lixing Zhu.
\newblock Model checking for parametric single-index models: a dimension
  reduction model-adaptive approach.
\newblock \emph{Journal of the Royal Statistical Society: Series B (Statistical
  Methodology)}, 78\penalty0 (5):\penalty0 1013--1035, 2016.

\bibitem[Hall et~al.(1988)Hall, Cumminq, Laderman, and Mundy]{hall1988r}
Bronwyn~H Hall, Clint Cumminq, Elizabeth~S Laderman, and Joy Mundy.
\newblock The r\&d master file documentation, 1988.

\bibitem[Heller et~al.(2012)Heller, Heller, and Gorfine]{heller2012consistent}
Ruth Heller, Yair Heller, and Malka Gorfine.
\newblock A consistent multivariate test of association based on ranks of
  distances.
\newblock \emph{Biometrika}, 100\penalty0 (2):\penalty0 503--510, 2012.

\bibitem[Jiang et~al.(2020)Jiang, Zhao, Ma, and Fan]{jiang2020}
Hangjin Jiang, Xingqiu Zhao, Ronald~C.W Ma, and Xiaodan Fan.
\newblock Consistent screening procedures in high-dimensional binary
  classification.
\newblock \emph{Statistica Sinica to appear}, 2020.

\bibitem[Khmaladze et~al.(2004)Khmaladze, Koul,
  et~al.]{khmaladze2004martingale}
Estate~V Khmaladze, Hira~L Koul, et~al.
\newblock Martingale transforms goodness-of-fit tests in regression models.
\newblock \emph{The Annals of Statistics}, 32\penalty0 (3):\penalty0 995--1034,
  2004.

\bibitem[Kim et~al.(2018)Kim, Balakrishnan, and Wasserman]{kim2018robust}
Ilmun Kim, Sivaraman Balakrishnan, and Larry Wasserman.
\newblock Robust multivariate nonparametric tests via projection-pursuit.
\newblock \emph{arXiv preprint arXiv:1803.00715}, 2018.

\bibitem[McCullagh(1989)]{mccullagh2018generalized}
Peter McCullagh.
\newblock \emph{Generalized linear models}.
\newblock Routledge, 1989.

\bibitem[Neumeyer and Van~Keilegom(2010)]{neumeyer2010estimating}
Natalie Neumeyer and Ingrid Van~Keilegom.
\newblock Estimating the error distribution in nonparametric multiple
  regression with applications to model testing.
\newblock \emph{Journal of Multivariate Analysis}, 101\penalty0 (5):\penalty0
  1067--1078, 2010.

\bibitem[Quinlan(1993)]{quinlan1993combining}
J~Ross Quinlan.
\newblock Combining instance-based and model-based learning.
\newblock In \emph{Proceedings of the tenth international conference on machine
  learning}, pages 236--243, 1993.

\bibitem[Stute et~al.(1998{\natexlab{a}})Stute, Manteiga, and
  Quindimil]{stute1998bootstrap}
W~Stute, W~Gonz{\'a}lez Manteiga, and M~Presedo Quindimil.
\newblock Bootstrap approximations in model checks for regression.
\newblock \emph{Journal of the American Statistical Association}, 93\penalty0
  (441):\penalty0 141--149, 1998{\natexlab{a}}.

\bibitem[Stute(1997)]{stute1997nonparametric}
Winfried Stute.
\newblock Nonparametric model checks for regression.
\newblock \emph{The Annals of Statistics}, pages 613--641, 1997.

\bibitem[Stute and Zhu(2002)]{stute2002model}
Winfried Stute and LI-XING Zhu.
\newblock Model checks for generalized linear models.
\newblock \emph{Scandinavian Journal of Statistics}, 29\penalty0 (3):\penalty0
  535--545, 2002.

\bibitem[Stute et~al.(1998{\natexlab{b}})Stute, Thies, Zhu,
  et~al.]{stute1998model}
Winfried Stute, Silke Thies, Li-Xing Zhu, et~al.
\newblock Model checks for regression: an innovation process approach.
\newblock \emph{The Annals of Statistics}, 26\penalty0 (5):\penalty0
  1916--1934, 1998{\natexlab{b}}.

\bibitem[Thas(2010)]{thas2010comparing}
Olivier Thas.
\newblock \emph{Comparing distributions}.
\newblock Springer, 2010.

\bibitem[Tsay(1992)]{tsay1992model}
Ruey~S Tsay.
\newblock Model checking via parametric bootstraps in time series analysis.
\newblock \emph{Journal of the Royal Statistical Society: Series C (Applied
  Statistics)}, 41\penalty0 (1):\penalty0 1--15, 1992.

\bibitem[Van~Keilegom et~al.(2008)Van~Keilegom, Manteiga, and
  Sellero]{van2008goodness}
Ingrid Van~Keilegom, Wenceslao~Gonz{\'a}lez Manteiga, and C{\'e}sar~S{\'a}nchez
  Sellero.
\newblock Goodness-of-fit tests in parametric regression based on the
  estimation of the error distribution.
\newblock \emph{Test}, 17\penalty0 (2):\penalty0 401--415, 2008.

\bibitem[Xia(2009)]{xia2009model}
Yingcun Xia.
\newblock Model checking in regression via dimension reduction.
\newblock \emph{Biometrika}, 96\penalty0 (1):\penalty0 133--148, 2009.

\bibitem[Zhang and Wu(2011)]{zhang2011testing}
Ting Zhang and Wei~Biao Wu.
\newblock Testing parametric assumptions of trends of a nonstationary time
  series.
\newblock \emph{Biometrika}, 98\penalty0 (3):\penalty0 599--614, 2011.

\bibitem[Zheng(1996)]{zheng1996consistent}
John~Xu Zheng.
\newblock A consistent test of functional form via nonparametric estimation
  techniques.
\newblock \emph{Journal of Econometrics}, 75\penalty0 (2):\penalty0 263--289,
  1996.

\bibitem[Zhou et~al.(2017)Zhou, Zheng, and Zhang]{zhou2017two}
Wen-Xin Zhou, Chao Zheng, and Zhen Zhang.
\newblock Two-sample smooth tests for the equality of distributions.
\newblock \emph{Bernoulli}, 23\penalty0 (2):\penalty0 951--989, 2017.

\end{thebibliography}

\section*{Appendix}

\renewcommand\thesubsection{\Alph{subsection}.}
\renewcommand\thesubsection{A.\arabic{subsection}}
\renewcommand{\thefigure}{A.\arabic{figure}}
\renewcommand{\thetable}{A.\arabic{table}}
\renewcommand{\thetheorem}{A.\arabic{table}}
\renewcommand{\theequation}{A.\arabic{equation}}
\setcounter{table}{1}
\setcounter{figure}{1}
\setcounter{theorem}{1}

\subsection{Discussion on MAC statistic}

One may notice that MAC is sensitive to local discrepancy between two distributions.  However, it is insensitive to the global discrepancy.  To better understand this point, let's see an example. Assume that we have $n=10$ local statistics, $Z_1, \cdots, Z_n$ for the statistical hypothesis test problem in hand. Define $Z_{max}=\max_{1\leq i \leq n} Z_i$, and we will reject the null hypothesis if, say, $Z_{max}>5$.  If we use the mean of these local statistics, $Z_{mean}=\frac{1}{n}\sum_{1\leq i \leq n} Z_i$, we will reject the null hypothesis if, say, $Z_{mean}>2$.  Now, we consider two different cases. (C1) the observation of $\{Z_1, \cdots, Z_{10}\}$ is $Z_i=4.5$ for $1\leq i \leq 10$.  In this case, $Z_{mean}$ rejects the null but $Z_{max}$ accepts it.  (C2) the observation of $\{Z_1, \cdots, Z_{10}\}$ is $Z_i=1$ for $1 \leq i \leq 9$ and $Z_{10}=10$. In this case, $Z_{mean}$ accepts the null but $Z_{max}$ rejects it.  Thus, the mean focuses on the global difference and ignores the local difference, and the maximum does the opposite.  We need to combine the power of mean and maximum, and define the new statistic as
\begin{equation}\label{statistic}
T(\textbf{D}, \textbf{D}^*)= T_{mean}(\textbf{D}, \textbf{D}^*)+T_{max}(\textbf{D}, \textbf{D}^*)I(T_{max}(\textbf{D}, \textbf{D}^*)>\tau_n),
\end{equation}
where $T_{mean}(\textbf{D}, \textbf{D}^*)=\frac{2}{k(k-1)}\sum_{i,j} S(\textbf{L}_i,\textbf{L}_j)$ is the mean of all local statistics and $T_{max}(\textbf{D}, \textbf{D}^*)=\text{MAC}(\textbf{D}, \textbf{D}^*)$.
Based on Theorem~\ref{consist}, we know that  if $\tau_n =O(\sqrt{n\log n})$, the limit distribution of $T$ under $H_0$ is the same as that of $T_{mean}(\textbf{D}, \textbf{D}^*)$. 
The following result is a direct corollary of Theorem~\ref{consist}.
\begin{corollary}\label{consist2} Following above notations and assuming that $r_{ij}\geq \eta>0 $ for any $i$ and $j$,  and $\tau_n =O(\sqrt{n\log k})$ then for any $\epsilon>0$, let $\delta=\frac{4n\epsilon}{\eta}(\frac{1}{\eta}+2)$, we have
	(1) under $H_0$, $P(T(\textbf{D}, \textbf{D}^*)\leq O(\sqrt{n\log k}))\rightarrow 1$, as $n\rightarrow \infty$; (2) under $H_1$, $P(T(\textbf{D}, \textbf{D}^*)> O(n))\rightarrow 1$, as $n\rightarrow \infty$.
\end{corollary}

\textbf{Remark A.1.} \textit{According to results in Theorem~\ref{consist}, we can set $\tau_n$ in real applications as the empirical $(1-\alpha)\times$ 100\% quantile of MAC statistic, where $\alpha$ is usually set as 5\%, 1\% or 0.1\%.}

With the combination of mean and maximum as in equation (A.1), i.e, $Z_{mix}=Z_{mean}+Z_{max}I(Z_{max}>5)$, we will reject the null in both cases considered in the above example.  However, the difference between two distributions are local due to the constraint of the distribution function that the total probability is 1, and it is not necessary for us to define the test statistic in the form of~(\ref{statistic}), which also involves more computational time.

\subsection{Proofs}
\begin{lemma} \label{lA1}
	Suppose that $\textbf{U}_n=(U_{n,1},\cdots,U_{n,k})$ and $\textbf{V}_n=(V_{n,1},\cdots,V_{n,k})$ are independently sampled from multinominal distributions with parameters $(n, a_1, \cdots, a_k)$ and $(n, b_1,\cdots,b_k)$, respectively. Define the column vectors $\textbf{a}=(a_1,\cdots,a_k)$ and $\textbf{b}=(b_1,\cdots,b_k)$. Under the null hypothesis $H_0$: $\textbf{a}=\textbf{b}$, we have
	$T=\sum_{i=1}^k \sum_{i=1}^k\frac{(U_{n,i}-V_{n,i})^2}{N\hat c_i}\rightarrow
	\chi_{k-1}^2 $, where
	$\hat{\textbf{c}}=\frac{\textbf{U}_n+\textbf{V}_n}{2n}$ be the
	estimate of $\textbf{a}$ under $H_0$, and $k\geq 2$.
\end{lemma}
Proof: Let $\{X_t\}_{t=1}^n$ and $\{Y_t\}_{t=1}^n$ be the observed
category series of the count vectors $\textbf{U}_n$ and
$\textbf{V}_n$, respectively. So we have $P(X_t=i)=a_i$ and
$P(Y_t=i)=b_i$. Let $N=2n$, we define
$$
X_i^*=(X_{1,i}^*,\cdots,X_{N,i}^*)=(\delta_{i,X_1},\cdots,\delta_{i,X_n},0,\cdots,0)
$$
$$
Y_i^*=(Y_{1,i}^*,\cdots,Y_{N,i}^*)=(0,\cdots,0,\delta_{i,Y_1},\cdots,\delta_{i,Y_n})
$$
where $\delta_{i,j}=I(i=j)$.

Let
$Z_i=U_{n,i}-V_{n,i}=\sum_{k=1}^NX_{k,i}^*-\sum_{k=1}^NY_{k,i}^*$,
under $H_0$, we have $E(Z_i)=0, \text{Var}(Z_i)=2na_i(1-a_i),
\text{Cov}(Z_i,Z_j)=-2na_ia_j$. So, by the multivariate
central limit theorem, we have
$$
\frac{1}{\sqrt{2n}}\textbf{Z}=\frac{\textbf{U}_n-\textbf{V}_n}{\sqrt{2n}}\rightarrow_d
N(0,\Sigma),
$$
where $\sigma_{i,j}$, the $(i,j)$-th element of $\Sigma$, is equal
to $a_i(\delta_{ij}-a_j)$. Since $\hat{\textbf{c}} \rightarrow_p
\textbf{a}$, we have
$$
\frac{\textbf{U}_n-\textbf{V}_n}{\sqrt{2n\hat{\textbf{c}}}}\rightarrow_d
N(0,I_k-\sqrt{{\textbf{a}}}\sqrt{{\textbf{a}}}')
$$
Thus, we have $T\rightarrow \chi_{k-1}^2$, as $n\rightarrow +\infty$.

\textbf{Proof of Theorem 1}

Denote $\hat p_{ij}=P_{ij}/n$, $\hat q_{ij}=Q_{ij}/n$, and $\hat r_{ij}=\hat p_{ij}+\hat q_{ij}$ as the estimate of $p_{ij}$, $q_{ij}$ and $r_{ij}$, respectively.  Let $\Delta_{ij}=p_{ij}-q_{ij}$, and $\hat{\Delta}_{ij}=\hat p_{ij}-\hat q_{ij}$ be its estimate.

It is well known that for $\hat p_{ij}$, $\hat q_{ij}$, we have
$$
P(|\hat p_{ij}-p_{ij}|>\epsilon) \leq 2e^{-n\epsilon^2/2}, \ \ P(|\hat q_{ij}-q_{ij}|>\epsilon) \leq 2e^{-n\epsilon^2/2}.
$$
Thus, we have,
$$
P(|\hat{\Delta}_{ij}-\Delta_{ij}|>\epsilon) \leq 4e^{-n\epsilon^2/8}, \ \ P(|\hat r_{ij}-r_{ij}|>\epsilon) \leq 4e^{-n\epsilon^2/8}.
$$
For convenience, let $S$ be the local statistic defined in equation (\ref{local}), and $S^p$ be its population version defined in equation (\ref{plocal}).
\begin{eqnarray*}
	|S-S^p|&=&n|\sum_{i=1}^2\sum_{j=1}^2 \frac{(\hat p_{ij}-\hat q_{ij})^2}{\hat r_{ij}}-\sum_{i=1}^2\sum_{j=1}^2 \frac{(p_{ij}-q_{ij})^2}{r_{ij}}| \\
	&\leq&n\sum_{i=1}^2\sum_{j=1}^2|\frac{(\hat p_{ij}-\hat q_{ij})^2}{\hat r_{ij}}- \frac{(p_{ij}-q_{ij})^2}{r_{ij}}| \\
	&=&n\sum_{i=1}^2\sum_{j=1}^2|\frac{(\hat p_{ij}-\hat q_{ij})^2}{\hat r_{ij}}-\frac{(\hat p_{ij}-\hat q_{ij})^2}{r_{ij}}+\frac{(\hat p_{ij}-\hat q_{ij})^2}{r_{ij}}- \frac{(p_{ij}-q_{ij})^2}{r_{ij}}| \\
	&\leq& n\sum_{i=1}^2\sum_{j=1}^2 \frac{\hat{\Delta}_{ij}^2}{\hat r_{ij}r_{ij}}|\hat r_{ij}-r_{ij}|+\frac{|\hat{\Delta}_{ij}^2-\Delta_{ij}^2|}{r_{ij}} \\
	&\leq& n\sum_{i=1}^2\sum_{j=1}^2 \frac{|1-r_{ij}/\hat r_{ij}|}{\eta}+\frac{2|\hat{\Delta}_{ij}-\Delta_{ij}|}{\eta}
\end{eqnarray*}
Define event $E_1=\{|\hat{\Delta}_{ij}-\Delta_{ij}|>\epsilon\}$ and $E_2=\{|\hat r_{ij}-r_{ij}|>\epsilon\}$. Now, on $E_1^c \bigcap E_2^c$, we have,
\begin{equation}
|S-S^p|\leq \frac{4n\epsilon}{\eta}(\frac{1}{\eta}+2)
\end{equation}
Finally, let $\delta=\frac{4n\epsilon}{\eta}(\frac{1}{\eta}+2)$, we have
\begin{equation}
P(|S-S^p|>\delta)\leq 1- P(E_1^c \bigcap E_2^c)\leq P(E_1)+P(E_2)=8e^{-n\epsilon^2/8}.
\end{equation}
Thus, $	P(|{\rm MAC}(\textbf{W}, \textbf{V})-{\rm MAC}^p(\textbf{W}, \textbf{V})|>\delta)\leq 16n^2e^{-n\epsilon^2/8}$, and the other two conclusions follow directly.
\end{document}